\documentclass[journal,onecolumn]{IEEEtran}

\usepackage{microtype}
\usepackage{graphicx}
\usepackage{subcaption}
\usepackage{booktabs} 
\usepackage{url}
\usepackage{amsfonts}
\usepackage{amssymb}
\usepackage{amsmath,bm,amsthm}
\usepackage{braket}
\usepackage{algorithmic}
\usepackage{algorithm}
\usepackage{xcolor}
\usepackage{float}

\newtheorem{theorem}{Theorem}
\newtheorem{lemma}{Lemma}
\newtheorem{proposition}{Proposition}
\newtheorem{corollary}{Corollary}
\newtheorem{example}{Example}
\newtheorem*{remark}{Remark}

\DeclareMathOperator*{\argmin}{arg\,min}
\DeclareMathOperator*{\argmax}{arg\,max}
\newcommand{\defin}{\overset{\Delta}{=}}
\newcommand{\cFd}{\mathcal{F}_d}
\newcommand{\bs}{\mathbf{s}}
\newcommand{\bxs}{\mathbf{x}^*}
\newcommand{\bphi}{{\mathbf{\Phi}}}
\newcommand{\bpsi}{{\mathbf{\Psi}}}
\newcommand{\br}{\mathbf{r}}

 \newcommand{\bss}{\bs^{*}}
\newcommand{\cE}{\mathcal{E}}
\newcommand{\cF}{\mathcal{F}}
\newcommand{\ba}{\mathbf{a}}
\newcommand{\bb}{\mathbf{b}}
\newcommand{\bu}{\mathbf{u}}
\newcommand{\bx}{\mathbf{x}}
\newcommand{\bl}{\bm{l}}
\newcommand{\by}{{\mathbf{y}}}
\newcommand{\bz}{{\mathbf{z}}}

\newcommand{\bq}{\mathbf{q}}
\newcommand{\bp}{\mathbf{p}}
\newcommand{\hbx}{\hat{\bx}}
\newcommand{\cL}{{\mathcal{L}}}
\newcommand{\cR}{{\mathcal{R}}}
\newcommand{\cO}{{\mathcal{O}}}
\newcommand{\sT}{\mathsf{T}}
\newcommand{\lsr}[2]{{#1}^{\perp{#2}}}

\newcommand{\defeq}{\overset{\text{def}}{=}}

\begin{document}

%
\title{Error-Correction for Sparse \\ Support Recovery Algorithms}


%
%
%

\author{Mohammad~Mehrabi and
     Aslan~Tchamkerten 
\thanks{M. Mehrabi is with the Data Sciences and Operations Department, Marshall School of Business, University of Southern California, Los Angeles, CA 90089-0809, US.}
\thanks{A. Tchamkerten is with the Department of Communications and Electronics, Telecom Paris, Institut Polytechnique de Paris, 19 Place Marguerite Perey, 91120 Palaiseau, FR.}
}

\maketitle

\begin{abstract}

	Consider the compressed sensing setup where the support $\bss$ of an $m$-sparse $d$-dimensional signal $\bx$ is to be recovered from $n$ linear measurements with a given algorithm. Suppose that the measurements are such that the algorithm does not guarantee perfect support recovery and that true features may be missed. Can they efficiently be retrieved? 
	
	This paper addresses this question through a simple error-correction module referred to as LiRE. LiRE takes as input an estimate $\bs_{\text{in}}$ of the true support $\bss$, and outputs a refined support estimate $\bs_{\text{out}}$.  In the noiseless measurement setup, sufficient conditions are established under which LiRE is guaranteed to recover the entire support, that is $\bs_{\text{out}}\supseteq \bss$. These conditions imply, for instance, that in the high-dimensional regime LiRE can correct a sublinear in $m$ number of errors made by Orthogonal Matching Pursuit (OMP). The computational complexity of LiRE is ${\cal{O}}(m n d)$.
	
	Experimental results with random Gaussian design matrices show that LiRE substantially reduces the number of measurements needed for perfect support recovery via Compressive Sampling Matching Pursuit, Basis Pursuit (BP), and OMP. Interestingly, adding LiRE to OMP yields a support recovery procedure that is more accurate and significantly faster than BP. This observation carries over in the noisy measurement setup where the combination of LiRE and OMP is faster and more accurate than LASSO. 
	Finally, as a standalone support recovery algorithm with a random initialization, experiments show that LiRE's reconstruction performance lies between OMP and BP.
	
	These results suggest that LiRE may be used generically, on top of any suboptimal baseline support recovery algorithm, to improve support recovery or to operate with a smaller number of measurements, at the cost of a relatively small computational overhead. Alternatively, LiRE may be used as a standalone support recovery algorithm that is competitive with respect to OMP.

\end{abstract}

\begin{IEEEkeywords}
 Compressed sensing, error-correction, feature selection, high dimension, linear model, support recovery
\end{IEEEkeywords}

%
\IEEEpeerreviewmaketitle
\section{Introduction}

Consider the support recovery problem of an $m$-sparse signal $\bx^*\in {\mathbb{R}}^d$ from $n<d$ linear measurements: \begin{align}\label{compsens}
    \by=\bphi \bx^*,
    \end{align}
where $\bphi\in {\mathbb{R}}^{n\times d}$ refers to the design matrix. Suppose an algorithm $A$ is used to recover the support $\bss$ of $\bx^*$ in a regime where errors may happen; for instance, the Restricted Isometry Property (RIP) constant of $\bphi$ need not be small enough to guarantee perfect support recovery---see, {\it{e.g.}}, \cite{wen2016sharp} for such a condition for OMP.  Can potential errors be efficiently corrected?

We address this question through a simple low complexity error-correction module, referred to as LiRE---for List Regression Error-correction. LiRE takes as input an initial support estimate $$\bs_{\text{in}}\defeq A(\by,\bphi)$$ provided by a baseline algorithm $A$, and produces a second support estimate $$\bs_{\text{out}}\defeq\text{LiRE}(\by,\bs_{\text{in}},\bphi)\overset{\text{notation}}{=}\text{LiRE}\circ A (\by,\bphi)$$ of size $m$. Under certain RIP conditions that depend on the number of missed features $|\bss\backslash \bs_{\text{in}}|$, LiRE's estimate $\bs_{\text{out}}$ includes $\bss$ (see Theorem~\ref{thm1} and Corollaries~\ref{corollary1},\ref{corollary2} in Section~\ref{mainres}). In the high-dimensional regime, these conditions imply, for instance, that we can use OMP in a regime where perfect support recovery is not guaranteed, and yet LiRE will recover all missed features as long as their number grows sublinearly in the sparsity level.

In a second part of the paper, we further assess the performance of LiRE first as an error-correction module, then as a standalone support recovery algorithm with a random initialization. We present four sets of numerical experiments
that address the following questions:  Can LiRE improve the support recovery of (good) baseline algorithms?
Can LiRE be combined with a baseline algorithm of similarly low complexity to achieve the performance of more complex reconstruction algorithms? 
Is LiRE robust to noise? And, 
how efficient is LiRE as a standalone support recovery algorithm? These questions are addressed by considering random Gaussian supports and design matrices.
\begin{itemize}
    \item The first experiment considers  OMP, BP (\cite{BP}), and Compressive Sampling Matching Pursuit (CoSaMP, \cite{cosamp}) as baseline algorithms. Results show that a few (five) iterations of LiRE increases and never decreases the average percentage of exact support recovery, for a non-trivial range of undersampling-sparsity operating regimes. In particular, LiRE reduces the number of measurements needed for perfect support recovery via CoSaMP, BP, and OMP by up to $15\%$, $25\%$, and $40\%$, respectively, depending on the the sparsity level.

\item 
The second set of simulations compares LiRE$\circ$OMP against BP and shows that even though LiRE$\circ$OMP has a significantly lower complexity than BP ($\cO(mnd)$ vs. $\cO(n^2d^{1.5})$, see \cite{new_fornasier2010numerical}), it achieves an average percentage of successful support recovery that is at least as large as BP, and sometimes larger. 
 
\item The third set of simulations evaluates the robustness of LiRE against noise in the Gaussian additive model $$
    \by=\bphi \bxs +\bz.
    $$ By repeating the second set of experiments but now with LiRE$\circ$OMP against the LASSO solution, we observe that LiRE$\circ$OMP is superior to LASSO as the noise level increases even though computing the LASSO has a computational cost that is at least quadratic in~$d$ \cite{efron2004least}. 
    \item The fourth set of simulations evaluates the performance of LiRE as a standalone support recovery algorithm with a random initialization. Results show that in terms of percentage of exact support recovery LiRE lies between OMP and BP.
\end{itemize}

\subsection{Related works}\label{relwo}
The problem of solving the under-determined system of equations given by \eqref{compsens} to recover the planted solution $\bxs$, a.k.a. compressed sensing, has a vast literature. By assuming some structural properties of $\bx^*$, there exists a unique solution to \eqref{compsens} which may be found efficiently depending on $\bphi$. The most common property is that $\bx^*$ is $m-$sparse. Finding the sparsest solution 
$$\min_{\bx:\bphi \bx=\by} ||\bx||_0$$ is a non-convex NP-hard problem~\cite{natarajan1995sparse}, but convex optimization can recover $\bx^*$ if the design matrix $\bphi$ satisfies certain conditions. For instance, in \cite{RIP} it is shown that if $\bphi$ satisfies a certain RIP condition, then the sparsest solution corresponds to $\bxs$ but also corresponds to the solution of the convex optimization problem $\min_{\bx:\bphi \bx=\by} ||\bx||_1$, a.k.a. Basis Pursuit. For alternative structural properties and related conditions see, {\it{e.g.}},  \cite{negahban2012unified,agarwal2010fast,bora2017compressed}. 

Over the past fifteen years, a significant amount of work has gone  into the design of ever more efficient reconstruction algorithms, and conditions on $\bphi$ under which reconstruction is possible, see, {\it{e.g.}}, \cite{BP,OMP,cosamp,1, blumensath2009iterative,foucart2011hard,4,khanna2017iht,new_shen2018least}. The performance of these algorithms is typically quantified in terms of computational complexity and the sparsity-undersampling tradeoff, that is the $n$ vs. $m$ curve (at fixed $d$) that characterizes the regimes where support recovery is attained, possibly within some prescribed distortion  \cite{incoherence,RIP,RSC,somani2018support,jain2017partial,zhao2019analysis,foucart2019iterative,aeron2010information,6142088}. Recall that in high dimensions (say, $d=\omega(n^3)$), the computational bottleneck for computing $\bxs$ lies in finding $\bss$. Given $\bss$, the signal $\bx^*$ is obtained by minimizing $||\by-\bphi \bx ||_2$ over all $\bx$ with support $\bss$. In turn, this least square estimate is equal to Penrose's pseudo-inverse $\bphi^{\dagger}$ of $\bphi$ applied to $\by$, which can be performed with order $\cO(m^2n)$ computations using direct methods ({\it{e.g.}}, QR factorization) or cost $\cO(mn)$ using approximation methods, {\it{e.g.}}, Richardson's method (see, {\it{e.g.}}, \cite[Section 5.1]{cosamp}).

There are two main categories of signal reconstruction algorithms. In the first category, algorithms attempt to solve a convex relaxation of the original non-convex optimization problem, similarly to BP. In the second category, a support estimate is first constructed, typically in an iterative and greedy manner, then the signal is estimated. Prominent algorithms here include OMP and CoSaMP. 
It is generally accepted that algorithms based on convex relaxations achieve better sparsity-undersampling performance than greedy algorithms at the cost of an increase in computational complexity---see, for instance,  \cite{new_donoho2012sparse} for complexity/performance comparison of BP and OMP. Beyond greedy algorithms and in the quest of ever faster reconstruction algorithms, a line of works takes advantage of  distributed and parallel computation to further reduce computational time (see, {\it{e.g.}}, \cite{mirzasoleiman2013distributed,khanna2017scalable}).

Although many recovery algorithms
have been shown to perform well in certain settings, the tradeoff between reconstruction performance
and computational complexity remains elusive in general. The present work is a further exploration of this tradeoff by providing a means to increase reconstruction performance at the cost of a relatively small computational overhead. 

\begin{remark}
This paper is an extended version of the ISIT $2021$ conference submission \cite{mehrabitchamkertenisit}. The main difference with the present paper is that the ISIT submission states Theorem~\ref{thm1} but without proof. The proof consists of a sequence of eight Lemmas and one proposition which shed light on the role of the list, the key component of LiRE. The ISIT version did not include Corollary~\ref{corollary1} and Example~\ref{example}, and included only part of the simulations presented here. In particular, simulations pertaining to LiRE's robustness to noise and pertaining to LiRE as a standalone support recovery algorithm are not present in the ISIT submission.
\end{remark}

\subsection{Paper organization}
We end this section with notational conventions. In Section~\ref{mainres}, we introduce LiRE and state sufficient conditions under which LiRE corrects all errors made by the baseline algorithm. In Section~\ref{numsim}, we provide experimental results. In Section~\ref{canalysis}, we prove the  results of Section~\ref{mainres}, and in Section~\ref{crem} we draw concluding remarks and outline open problems. 

\subsection{Notational conventions }
The set of possible features is  denoted as $\cFd \defeq\{1,2,...,d\}$. A \textit{support vector} refers to a vector whose entries in $ \cFd$ are listed the \textit{ascending order}. Given a support vector $\bs$, $\bs[i]$ denotes the $i$th entry of $\bs$, $\bs[-i]$ denotes the vector obtained by removing the $i$th entry of $\bs$, and $\bar{\bs}$ denotes a vector with entries in $\cFd$ and not in $\bs$. We use $j \in \bs$ whenever $\bs[i]=j$ for some $i\in \{1,2,...,d\}$. The length and the support of a vector $\bx\in \mathbb{R}^d$ are denoted as $ |\bx|$ and $\text{supp}(\bx)$, respectively. Vector $ \bx$ is said to be $m$-sparse if $|\text{supp}(\bx)|\leq m$. Given support vectors $\br$ and $\bs$, we use $\br\backslash \bs$ to denote the support vector whose entries belong to $\br$ and not to $\bs$. We write $\br \subset \bs$ whenever $\br\backslash \bs$ is the null vector. Further, we use $\br \cap \bs$ to denote the support vector whose entries appear in both $\br$ and $\bs$, and use $\br \cup \bs$ to denote the support vector whose entries appear in $\bs$ or $\br$. 
For example, given support vectors $\br= [1,2,3,4]$ and $ \bs=[1,4,5]$ in  $\mathcal{F}_6$, we have $\br[-1]=[2,3,4], \bs\backslash \br=[5], \bs \cap \br=[1,4]$, $\br \cup \bs=[1,2,3,4,5]$, and that $\br[-1]$ and $ \bs\backslash\br$ are disjoint. 

Throughout the paper $\bphi$ refers to an $n\times d$ real matrix and we use $\bphi_{\bs}$ to denote the matrix $\bphi$ restricted to the set of columns indexed by the entries of $\bs$. The transpose of matrix $\bphi$ is denoted as $\bphi^\sT $ and $\bphi_\bs^\sT $ denotes $(\bphi_\bs)^\sT $.

Given $\bpsi\in \mathbb{R}^{n\times m}$, the (least square) residual of $\by\in \mathbb{R}^n$ is defined as
\begin{align*}
    \lsr{\by}{\bpsi} \defeq \by-\bpsi\cdot \argmin\limits_{\bz \in \mathbb{R}^m}||\by-\bpsi \bz ||_2. 
\end{align*}
Recall that if $\bpsi^\sT \bpsi$ is invertible then $$\lsr{\by}{\bpsi}=(I-P\{ \bpsi\})\by$$ where $$P\{\bpsi\}\defeq \bpsi(\bpsi^\sT \bpsi)^{-1}\bpsi^\sT $$ is the projection operator (see, {\it{e.g.}}, \cite{boyd2004convex}).

With a slight abuse of notation, the residual with respect to a support vector $\bs$ is defined as $$\lsr{\by}{\bs}\defeq\lsr{\by}{\bphi_{\bs}}.$$

The support vector consisting of the $\ell$ most correlated features with respect to $\lsr{\by}{\bs}$ is defined as
$$ \cL(\ell,\bs )\defeq\argmax\limits_{\bq, |\bq|=\ell}^{} ||\bphi_\bq^\sT  \by^{\perp \bs}||_1 \,, $$
where the maximum is intended to be over support vectors in $\cFd$.
The least square estimate with respect to a support vector $\bs$ is defined as 
$$\cE(\bs)\defeq \argmin\limits_{\bx \in \mathbb{R}^d, \text{supp}(\bx) \subseteq \bs}^{}||\by-\bphi\bx||_2\,.$$

\section{List Regression Error-Correction (LiRE)}\label{mainres}
Consider the noiseless linear model \begin{align} \label{model}
    \by=\bphi\bxs\end{align} with $\bxs \in \mathbb{R}^d$, $\by \in \mathbb{R}^{n}$, $n\leq d$, and where the design matrix $\bphi \in \mathbb{R}^{n\times d}$ is assumed to have unit $l_2$ columns without loss of generality. Given $\bphi$, $\by$, and knowing that $$|\bss|\defeq |\text{supp}(\bxs)|\leq m\leq n,$$ we want to find $\bs$ such that $\bss \subseteq \bs$ and $|\bs|=m$.

 Given an initial estimate $\bs_{\text{in}}$ of $\bs^*$ that potentially misses true features, LiRE attempts to produce a second estimate $\bs_{\text{out}}$ that contains all true features, that is   $ \bs_{\text{out}}  \supseteq\bs^*$. At the heart of LiRE is a leave-one-out procedure for error-correction purpose which checks, for each feature of the initial support estimate,  whether it can be replaced by a better one. The pseudo-code of LiRE is given below and is followed by comments in light of existing support recovery algorithms:


\begin{algorithm}[H]
   \caption{LiRE (List Regression Error-Correction) } 
   \label{alg1}
\begin{algorithmic}
   \STATE {\bfseries Input:} 
      \STATE $\bullet$ $n \times d$ real design matrix $\bphi$
   \STATE $\bullet$ $n$ dimensional data vector $\by$
   \STATE $\bullet$ upper bound $m$ on support size of $\bxs$  
   \STATE $\bullet$  list size $1\leq \ell\leq m$ (internal parameter)
   \STATE $\bullet$ initial estimate $\bs_{\text{in}}$ of $\bs^*$ with $|\bs_{\text{in}}|\leq m$
  \STATE {\bfseries Initialization:}
  \STATE   If $|\bs_{\text{in}}| < m$, add any $m-|\bs_{\text{in}}|$ features from $\cFd$ to $\bs_{\text{in}}$ such that $|\bs_{\text{in}}|=m$. Set $\bs_{\text{out}}=\bs_{\text{in}}$.
  \STATE {\bfseries Procedure:}
 \FOR{$i=1$ {\bfseries to} $m$}
\STATE   $\bullet$ If $\by^{\perp \bs_{\text{out}}[-i]}=0$, exit the {\bf{for}} loop. Else, find the $\ell$ most correlated features with respect to $\by^{\perp \bs_{\text{out}}[-i]}$:  
   $$ \bl\defeq\cL(\ell,\bs_{\text{out}}[-i]) .$$
\STATE  $\bullet$ Compute the least square estimate with respect to features $\bs_{\text{out}}[-i]\cup \bl $:
$$\hat{\bx}\defeq \cE(\bs_{\text{out}}[-i]\cup \bl)$$
\STATE $\bullet$ Pick $j \in \bl$  such that $|\hat{\bx}_j|=||\hat{\bx}_{\bl}||_{\infty}$. 
\STATE $\bullet$ Replace $\bs_{\text{out}}[i]$ by $j$. 
\ENDFOR
   \STATE{\bfseries Output:}
   \STATE $\bullet$ Support vector $\bs_{\text{out}}$
\end{algorithmic}
\end{algorithm}

Feature $i$ is first removed from the current support estimate and the residual is computed. If the residual is non-zero,  the $\ell$ features that are most correlated with the residual are added to the support. This results in an expanded support of size $m-1+\ell$, from which a signal estimate is computed. Finally, feature $i$ is replaced with the most relevant feature of this estimate, restricted to the list elements---in particular, feature $i$ could replace itself if it belongs to the list.

The support expansion through the list is reminiscent  of the signal proxy formation in several greedy algorithms, including OMP, gOMP \cite{wang2012generalized}, and CoSaMP. This step, however, serves here the purpose of error-correction as it allows to test whether a particular feature should be replaced or not. Intuitively, a wrong feature is less likely to be corrected if the list size is small. But a correct feature is also more likely to get replaced by a wrong feature if the list is large. Accordingly, the theoretical  guarantees for successful error-correction provided below (Section~\ref{difsuf}) tie the number of errors and the list size in an attempt to strike a balance between these two types of error.

\subsection{Computational complexity of LiRE} 

For each of its $m$ rounds, LiRE involves:
\begin{itemize}
    \item 
Two least square problems, cost $\cO(m^2n)$ using direct methods or cost $\cO(mn)$ using approximation methods (see Section~\ref{relwo}).
\item $d$ inner products of $n$ dimensional vectors, cost $\cO(nd)$.
\item Two sortings of $d$ numbers, cost $\cO(d\log d)$ ({\it{e.g.}}, by Merge Sort).
\end{itemize}
 Hence,  LiRE's computational cost is $\cO(mnd)$, with the restriction $m\leq \sqrt{d}$ if we use direct methods for the least square problems.
 Notice that the second and third computations can be performed efficiently through parallelization.

\subsection{Sufficient conditions for error-correction}\label{difsuf}
Theorem~\ref{thm1} below provides a sufficient condition under which one pass of LiRE recovers all missed true features of the initial support estimate. 
\begin{theorem}
\label{thm1}
Fix integers $1\leq \ell\leq m\leq n \leq d$, 
and consider the model \eqref{model} for a given design matrix $\bphi\in \mathbb{R}^{n\times d}$. Let $\bs_{\text{in}}$, with $|\bs_{\text{in}}|\leq m$, denote an estimate of the true support $\bs^*$, let $e\defeq |\bss \backslash  \bs_{\text{in}}|$ denote the number of missed true features, and let $$t\defeq\max\{m+e, \ell+e+1 \},$$
$$\eta_t\defeq \frac{\sqrt{2}\delta_t(1-\delta_t^2)}{(1-\delta_t-\delta_t^2)(1-2\delta_t)},$$
where $\delta_t$ denotes the order-$t$ RIP constant of $\bphi$.\footnote{Given an integer $t\geq 1$, the order-$t$ Restricted Isometry Property (RIP) constant $\delta_t$ of matrix $\bphi$ is defined as the smallest $\delta$ such that the inequality 
\begin{align*}
 (1-\delta)||\bx||_2^2 \leq ||\bphi\bx||_2^2\leq (1+\delta)||\bx||_2^2
\end{align*}
holds for all $t$-sparse vectors $\bx$. }

Suppose that $e$, $\ell$, and $\bphi$ satisfy the following inequalities:
\begin{align}\label{elcond}
  \ell\leq \max\{e,1\}\end{align}
\begin{align}\label{sqe}
\sqrt{e+1} \leq \frac{\sqrt{2}(1-\delta_t-\delta_t^2)(1-2\delta_t) }{\delta_t(1+\delta_t)(1+2\delta_t-\delta_t^2)}\end{align}
\begin{align}\label{sqe2}\sqrt{\ell} > \frac{(1-\delta_t^2+\delta_t)\eta_t\sqrt{e+1}-1+\delta_t}{1-\delta_t-\delta_t\eta_t\sqrt{e+1}}\sqrt{e+1}\end{align}
 \begin{align}\label{sqe3}
 \delta_{\ell+m-1}<0.5.\end{align}
 
 Then, given $\bs_{\text{in}}$, LiRE with list size $\ell$ outputs $\bs_{\text{out}}$ such that $\bs_{\text{out}} \supset \bss$ and $|\bs_{\text{out}}|= m$.
\end{theorem}
A more explicit sufficient condition for error correction is obtained by choosing  $\ell=\max\{e,1\}$, which implies $t\leq m+e+2$ since $e\leq m$, and $\ell\leq e+1$. Condition \eqref{sqe} in Theorem~\ref{thm1} with strict inequality implies Condition~\eqref{sqe2}. Furthermore, Condition~\eqref{sqe3} is implied by the condition $\delta_{m+e}< 0.5$ since $\ell\leq e+1$. It then follows: 
\begin{corollary}\label{corollary1} LiRE with list size $\ell=\max\{e,1\}$ corrects exactly $e$ errors ({\it{i.e.}}, $|\bs_{\text{out}}|=m$, $\bs_{\text{out}} \supset \bss$, $|\bss \backslash  \bs_{\text{in}}|= {e}$) if $\bphi$ satisfies the following RIP conditions:
$$\delta_{m+e}< 0.5,$$ 
\begin{align}\label{strine}e+1< \left(\frac{\sqrt{2}(1-\delta_{m+e+2}-\delta_{m+e+2}^2)(1-2\delta_{m+e+2}) }{\delta_{m+e+2}(1+2\delta_{m+e+2}-\delta_{m+e+2}^2)(1+\delta_{m+e+2})}\right)^2.\end{align}
\end{corollary}

Theorem~\ref{thm1} and Corollary~\ref{corollary1} assume that the number of errors to be corrected is known to be exactly $e$. If we want LiRE to correct {\it{any}} number of errors up to some number $\bar{e}\geq 1$, then by Condition \eqref{elcond} we should pick $\ell=1$ which yields the following result:



\begin{corollary}\label{corollary2}
LiRE with list size $\ell=1$ corrects up to $\bar{e}$ errors if $\bphi$ satisfies the following RIP conditions:
\begin{align}\label{deltm}\delta_m< 0.5,\end{align} 
\begin{align}\label{condcor}
     \bar{e}+1  \leq \left(\frac{(1-\delta_{m+\bar{e}+1}-\delta_{m+\bar{e}+1}^2)(1-2\delta_{m+\bar{e}+1})}{\delta_{m+\bar{e}+1}\sqrt{2}(1+\delta_{m+\bar{e}+1}-\delta_{m+\bar{e}+1}^2)(1+\delta_{m+\bar{e}+1})} \right)^2.
\end{align}
\end{corollary}
\begin{proof}[Proof of Corollary \ref{corollary2}]
Inequality \eqref{condcor} implies~\eqref{sqe} and renders inequality~\eqref{sqe2} vacuous as its right-hand side becomes non-positive. Finally, note that \eqref{sqe3} is satisfied if $\ell=1$ and \eqref{deltm} holds.
\end{proof}

For small values of $\delta$'s the upperbound \eqref{condcor} is about four times larger than the upperbound \eqref{strine}, but is sufficient to show that LiRE can actually correct errors beyond the regime of exact support recovery of certain baseline algorithms.  The following example shows that LiRE corrects a sublinear in $m$ number of errors in a regime where OMP may produce errors:
\begin{example}[OMP, sublinear number of errors]\label{example}
Corollary~\ref{corollary2} implies that LiRE corrects up to $\bar{e}$ errors if 
\begin{align}\label{ebarcondi}\delta_{m+1+\bar{e}}\leq \frac{(1+o(1))}{\sqrt{\bar{e}}}\qquad \text{as }\bar{e},m\to \infty \text{ with } \bar{e}\leq m.
\end{align}
On the other hand, the RIP necessary (and sufficient) condition for OMP to recover the support is (see \cite{7763885, wang2012recovery})\footnote{If \eqref{ompinequa} is reversed, then there exists design matrices for which OMP is not guaranteed to always recover $\bss$.} \begin{align}\label{ompinequa}\delta_{m+1}< \frac{1}{\sqrt{m+1}}.
\end{align}

Now, using the property $\delta_{a b}\leq b\cdot \delta_{2a}$ for positive integers $a$ and $b$ \cite[Corollary 3.4]{cosamp} with $a=(m+1)/2$ and $b=2(1+\bar{e}/(m+1))$, we deduce that in the  regime  $\bar{e},m\to \infty$ with $\bar{e}=o(m)$, Condition \eqref{ompinequa} implies
$$\delta_{m+1+\bar{e}} \leq \frac{2(1+o(1))}{\sqrt{m+1}},$$
which is is more stringent than \eqref{ebarcondi}. In summary, if the design matrix satisfies \eqref{ebarcondi}, OMP alone is not guaranteed to recover the support and LiRE will retrieve the missed features as long as their number is known and sublinear in the sparsity. 
\end{example}
In general, to quantify the benefits due to LiRE through Corollary~\ref{corollary2} we need to identify a meaningful regime, namely conditions on the design matrix under which the baseline algorithm potentially makes errors, and an upper bound on the number of errors. (Note that from Corollary~\ref{corollary2} it is unclear how LiRE performs when the actual number of errors is above $\bar{e}$ or different than $e$ for Theorem~\ref{thm1} and Corollary~\ref{corollary1}.)

Unfortunately, such conditions are hardly available. And even when they are, such as for OMP,  non-trivial upper bounds on the number of errors remain elusive.\footnote{Note that here we are interested in conditions on the design matrix under which a suboptimal algorithm produces at most a certain number of errors, in the worst-case over  $\bxs$. In fact, several works investigate the fundamental limitations of non-zero error support recovery, or recovery with distortion, in probabilistic setups (see, {\it{e.g.}}, \cite{aeron2010information,6142088}).} 
 Hence, to provide a practical assessment of the performance of LiRE without bounds on the number of errors we resorted to numerical simulations which are presented in the next section.

\section{Numerical experiments}\label{numsim}

\begin{figure*}
\begin{center}
\centerline{\includegraphics[scale=.7]{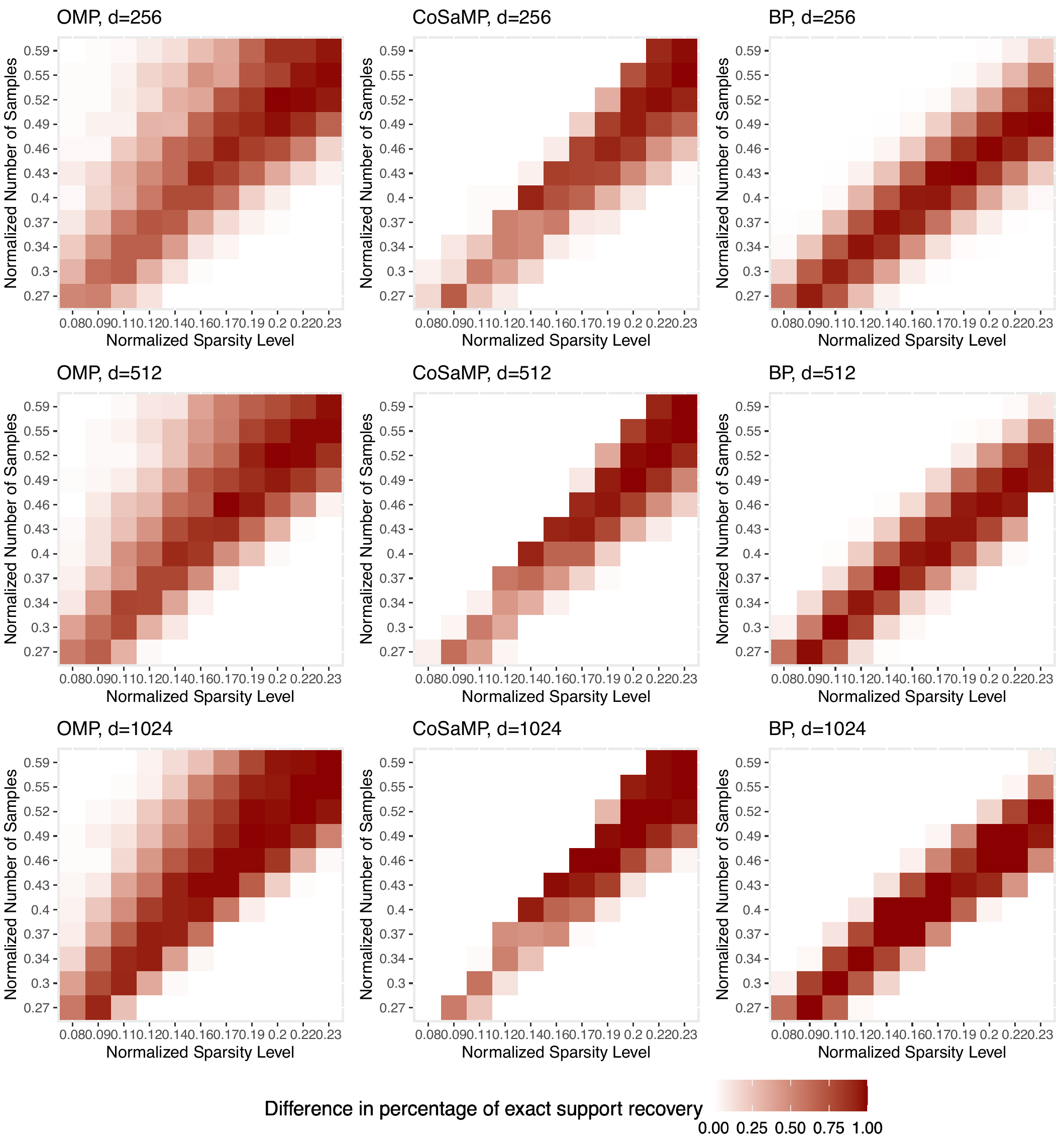}}
\caption{Improvement due to LiRE. Levels of red indicate the difference in the percentage of exact support recovery with and without LiRE, for OMP, CoSaMP, and BP---differences are always nonnegative. White cells indicate no improvement and dark red cells indicate regimes where LiRE allows to entirely recover the underlying support whereas the baseline algorithm alone makes errors.}\label{fig1}
\end{center}
\vskip -0.2in
\end{figure*}

The simulations described next evaluate the performance of LiRE both as an error correction module and as a standalone support recovery algorithm, with random Gaussian support and design matrices. Following the guidelines of reproducible research, the code for the simulations described below is availabe at:

\centerline{\url{https://github.com/mehrabi4/LiRE}.}
\subsection*{List size}
Theorem~\ref{thm1} suggests that the list size should be tied to the number of errors. Since this number is typically unknown, for all the numerical experiments we chose  $\ell$ empirically to be equal to $0.5m$ if $1.5m\leq n$, and $n-m$ otherwise.

\subsection{Improvement over OMP, CoSaMP, and BP}\label{improve}
 The first set of experiments shows how LiRE improves the support recovery of  OMP, CoSaMP, and BP for different normalized sparsity levels $m/d$ and number of samples $n/d$. Experiments were carried for $d\in \{256, 512, 1024\}$ for different values of $m$ and $n$ with step size of about $0.015d$ for $m$ and $0.03d$ for $n$. Given $m,n,d$, a design matrix $\bphi$ was first randomly generated with i.i.d. normal $\mathcal{N}(0,\frac{1}{n})$ entries. The signal $\bx$ was randomly generated with a uniformly chosen support of size $m$ and with independently generated i.i.d. $\mathcal{N}(0,1)$ entries. The percentages of perfect support recovery over $50$ experiments were computed with and without LiRE---each experiment is run with and without LiRE. For each experiment, LiRE was run $5$ consecutive times. The number of iterations of CoSaMP was chosen to be $d/4$.  In~Fig.~\ref{fig1}, levels of red indicate improvement---percentage of exact support recovery with LiRE minus percentage of exact support recovery without LiRE---from $0\%$ (white squares) to $100\%$ (dark red squares). \subsubsection*{Discussion} The lower white region in each plot of Fig.~\ref{fig1} corresponds to no recovery ($0\%$) with and without LiRE, whereas the upper white region corresponds to maximal ($100\%$) recovery with and without LiRE. Between these two regions LiRE improves recovery, up to $100\%$ for the dark red regions. We observe that LiRE never degrades performance, and numerical values at $d=1024$ reveal that LiRE reduces the number of samples needed to reach perfect ($100\%$) reconstruction by $30-40\%$ for OMP, $10-15\%$ for CoSaMP, and $15-25\%$ for BP. Notice that the dark red regions correspond to regimes where the baseline algorithm always produces errors, and LiRE correct them all.
 
The choice of running LiRE five times was based on empirical trials. This guaranteed a non-negative improvement of the percentage of exact support recovery across the entire range of sparsity levels (from $0.08$ to $0.23$). By contrast, we observed that fewer iterations could result in a slight decrease in performance for certain sparsity levels. This is possibly due to the fact that the list size was not optimized.

\subsection{LiRE$\circ$OMP versus BP}\label{bpLi}
The second set of experiments is perhaps the most interesting as it shows how LiRE may boost the performance of a low complexity reconstruction algorithm to achieve the performance of significantly more complex ones. We ran the same experiments as in the previous section but now compared BP against LiRE$\circ$OMP. Levels of red in Fig.~\ref{fig2} correspond to the percentage of exact support recovery of LiRE$\circ$OMP minus the percentage of exact support recovery of BP, with LiRE run $1$, $3$, and $5$ consecutive times. 

Fig.~\ref{fig22} shows the percentage of exact support recovery as a function of the number of measurements, for LiRE$\circ$OMP and BP, with LiRE run only once, $d=512$, and three sparsity levels ($35$, $70$, and $100$).

\begin{figure}
\begin{center}
\centerline{\includegraphics[scale=.7]{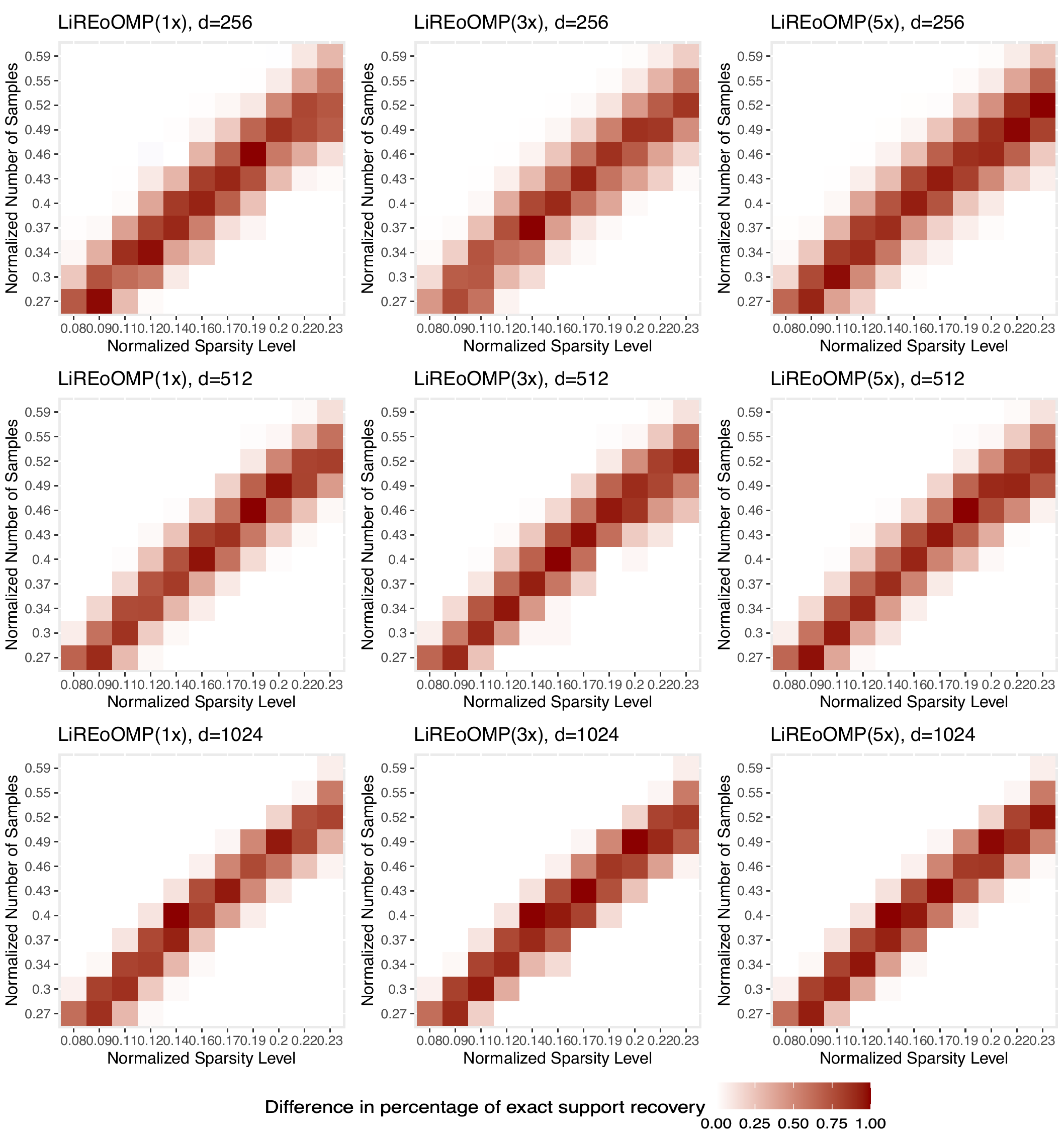}}
\caption{Same experiment as illustrated in Fig.~\ref{fig1}, but for  LiRE$\circ$OMP vs. BP, with $1, 3$ and $5$ iterations of LiRE.}\label{fig2} 
\end{center}
\vskip -0.2in
\end{figure}

 \begin{figure}
\begin{center}
\centerline{\includegraphics[scale=.5]{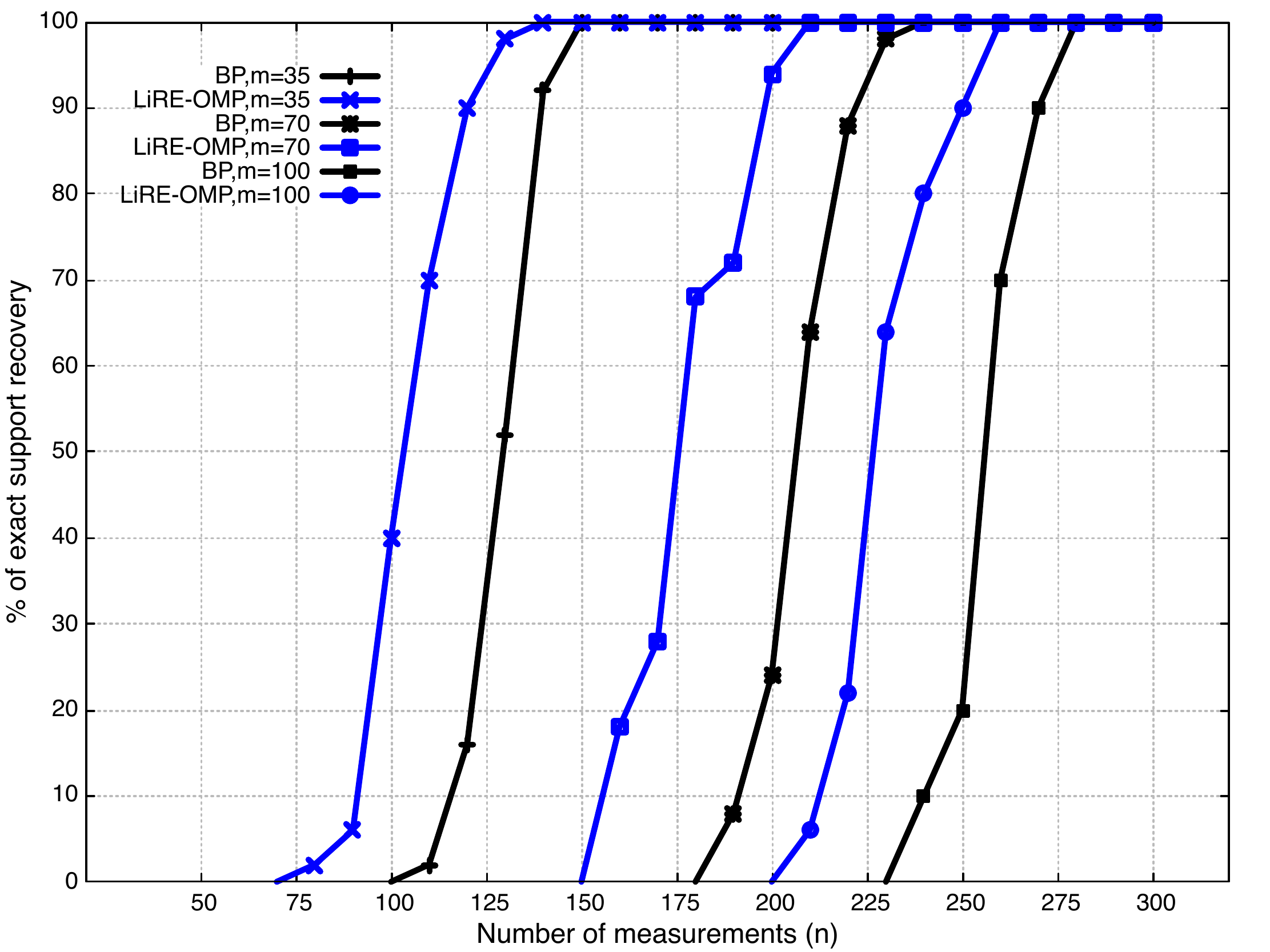}}
\caption{LiRE$\circ$OMP vs. BP, for $d=512$.}\label{fig22} 
\end{center}
\vskip -0.2in
\end{figure}

\subsubsection*{Discussion}
Referring to Fig.~\ref{fig2}, we observe that LiRE$\circ$OMP performs at least as well BP---in certain regimes strictly better---despite the significantly lower computational cost: $\cO(mnd)$ for LiRE$\circ$OMP vs. $\cO(n^2d^{1.5})$ for BP \cite{new_fornasier2010numerical}. We also note that the region of improvement (non-white region) does not substantially change with the number of iterations. Finally, the numerical data for Fig.~\ref{fig22} shows that LiRE$\circ$OMP reduces by $5-10\%$ the number of measurements needed by BP to achieve perfect support recovery. 
\subsection{Robustness to noise: LiRE$\circ$OMP vs. LASSO}
In the third set of experiments, see Fig.\ref{fig3}, we considered the noisy measurement model $$\by=\bphi \bx^* +\bz,$$ and compared LiRE$\circ$OMP (with LiRE run only once) against LASSO, similarly as for the previous set of experiments (difference of percentage of exact recovery computed over $50$ experiments for any given pair of normalized sparsity level and number of samples). The noise vector $\bz$ was  i.i.d. Gaussian across components, with zero mean and variance $\sigma^2\in \{0.0005, 0.001, 0.002\}$. The variance values were chosen empirically; high enough to put LASSO outside its comfort zone of perfect reconstruction, but not too high to allow LiRE to correct errors. The $\ell_1$-regularization parameter of LASSO was chosen by performing a ten-fold cross-validation.

\begin{figure}[h!]
\begin{center}
\centerline{\includegraphics[scale=.7]{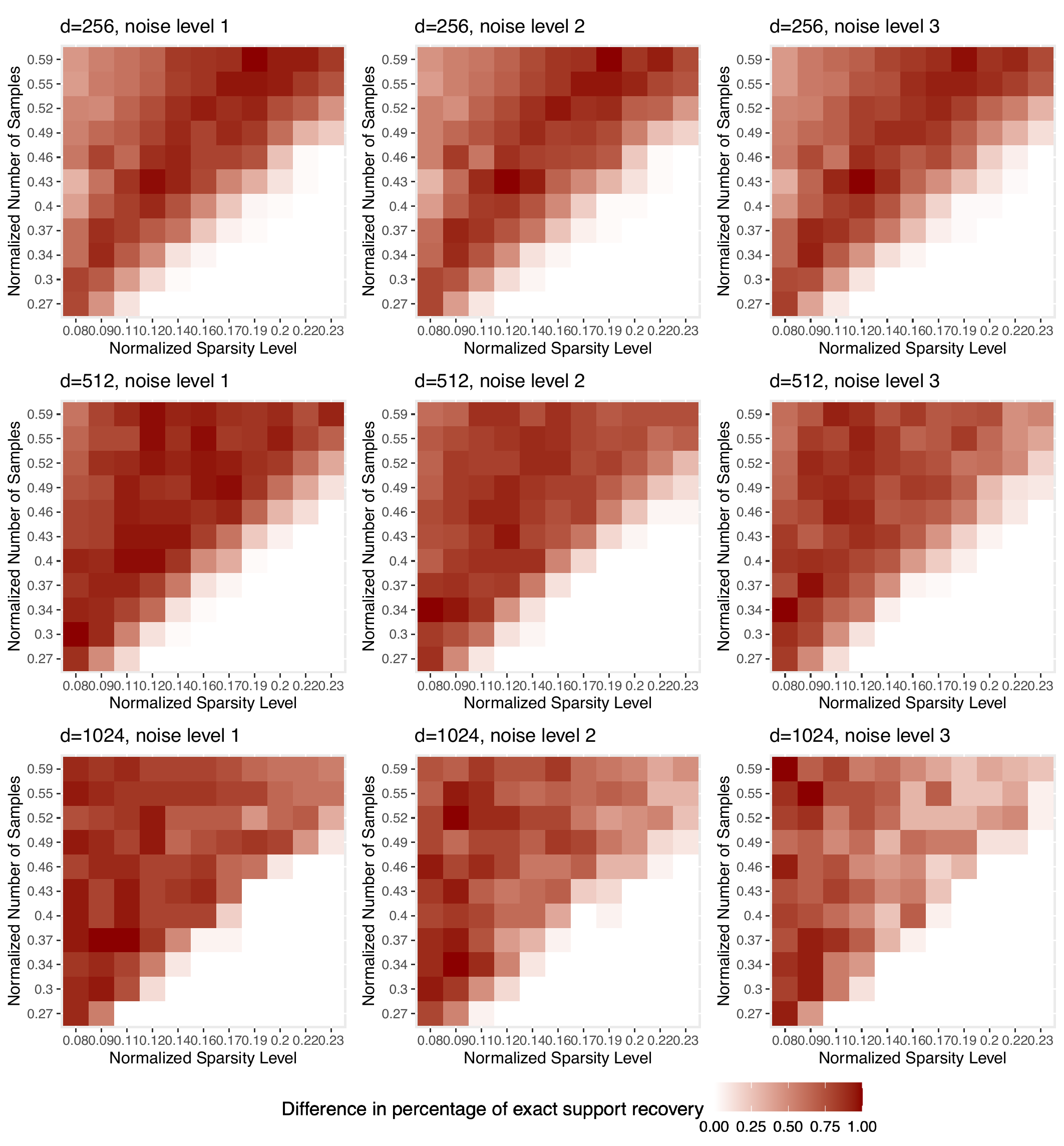}}
\caption{Same experiment as illustrated in Fig.~\ref{fig1}, but for LiRE$\circ$OMP vs. LASSO at different noise levels.}\label{fig3} 
\end{center}
\end{figure}

\subsection*{Discussion}
In Fig.~\ref{fig3}, we observe that  LiRE$\circ$OMP always improves support recovery over LASSO (except for the white region where the support is never fully recovered, with and without LiRE) and is also significantly faster since LASSO solver's complexity is quadratic or cubic in~$d$, depending on the sparsity level \cite{efron2004least}. Outside the white region the improvement is significant, mostly by more than $50\%$. 


\subsection{LiRE as a standalone support recovery algorithm?}\label{stalone}
In the fourth set of experiments (Fig.~\ref{fig:fig}), we compared LiRE (one pass) as a standalone support recovery algorithm against BP and OMP, and run similar experiments as described in Section~\ref{bpLi}. LiRE was initialized with a randomly and uniformly selected $\bs_{\text{in}}$.
Fig.~\ref{fig:sfig1} gives the percentage of exact support recovery with BP minus the percentage of exact support recovery with LiRE. Red favors BP and blue favors LiRE, and similarly for LiRE against OMP in Fig.~\ref{fig:sfig2} where red favors LiRE. Recall that the computational complexities of LiRE and OMP are similar, of order ${\cal{O}}(mnd)$, whereas it is order $\cO(n^2d^{1.5})$ for BP. 

\subsection*{Discussion}
Fig.~\ref{fig:sfig1} shows that BP is superior to LiRE at higher sparsity levels. At lower sparsity levels performances appear to be close. Fig.~\ref{fig:sfig2} shows that LiRE achieves better performance than OMP, and that the difference gets more pronounced as the sparsity level increases.

Note here that with a random initialization $\bs_{\text{in}}$ may contain no correct feature, that is $\bar{e}=m$. Applying Corollary~\ref{corollary2} with the trivial upper bound $\bar{e}=m$ yields the sufficient condition
\begin{align}\label{emlire}\delta_{2m+1}\leq \frac{1}{\sqrt{2 m}}(1+o(1)) \quad m\to \infty\end{align}
for LiRE to recover an arbitrary number of errors. This condition is significantly more stringent than the sufficient condition for OMP to recover the underlying support: $\delta_{m+1}< \frac{1}{\sqrt{m+1}}$. In light of Fig.~\ref{fig:sfig2}, LiRE's sufficient condition \eqref{emlire} to recover an arbitrary number of errors appears to be loose.

\begin{figure}
\begin{subfigure}{0.5\textwidth}
  \centering
  \includegraphics[scale=0.65]{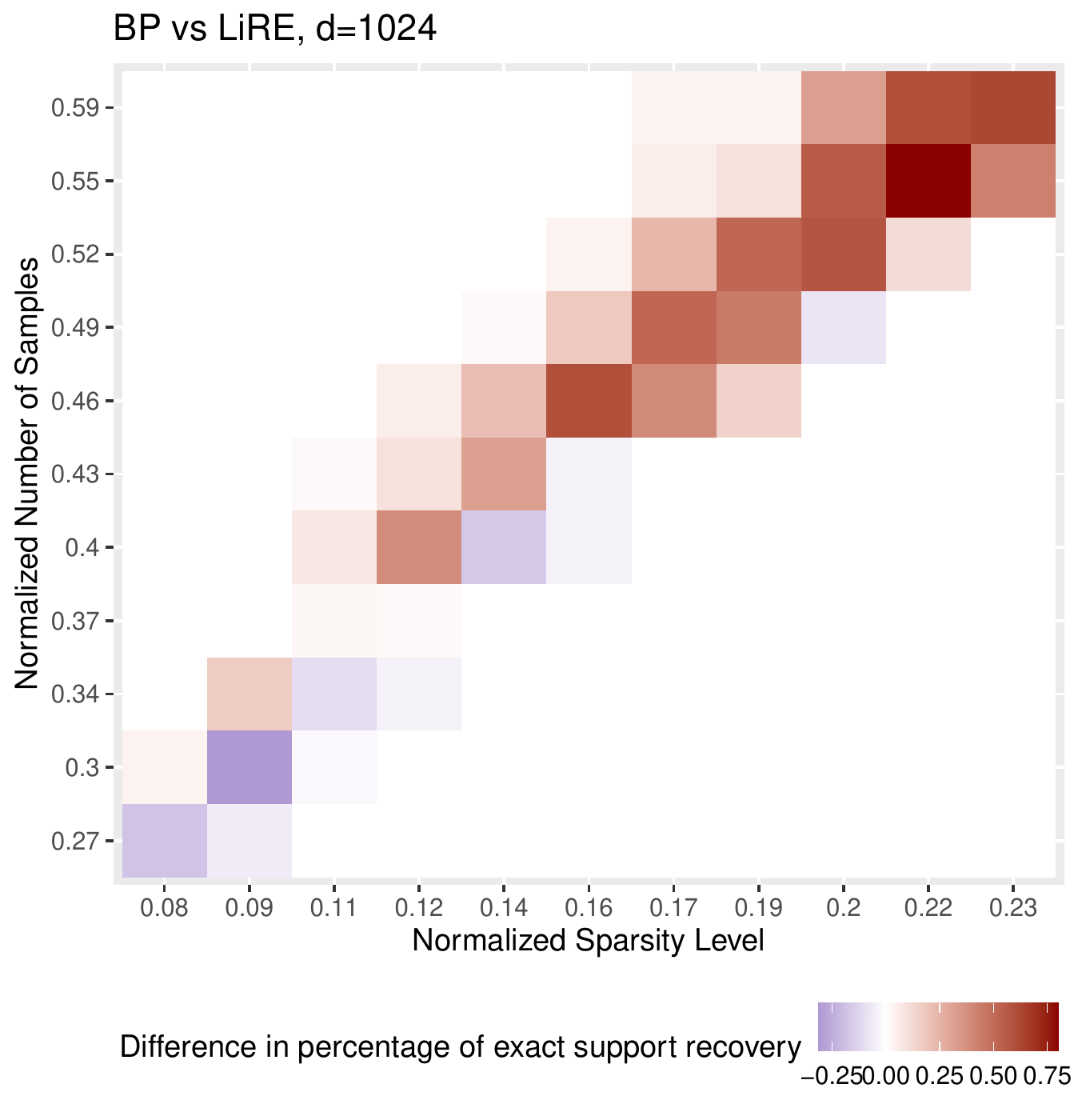}
  \caption{BP vs. LiRE.}
  \label{fig:sfig1}
\end{subfigure}%
\begin{subfigure}{0.5\textwidth}
  \centering
  \includegraphics[scale=0.65]{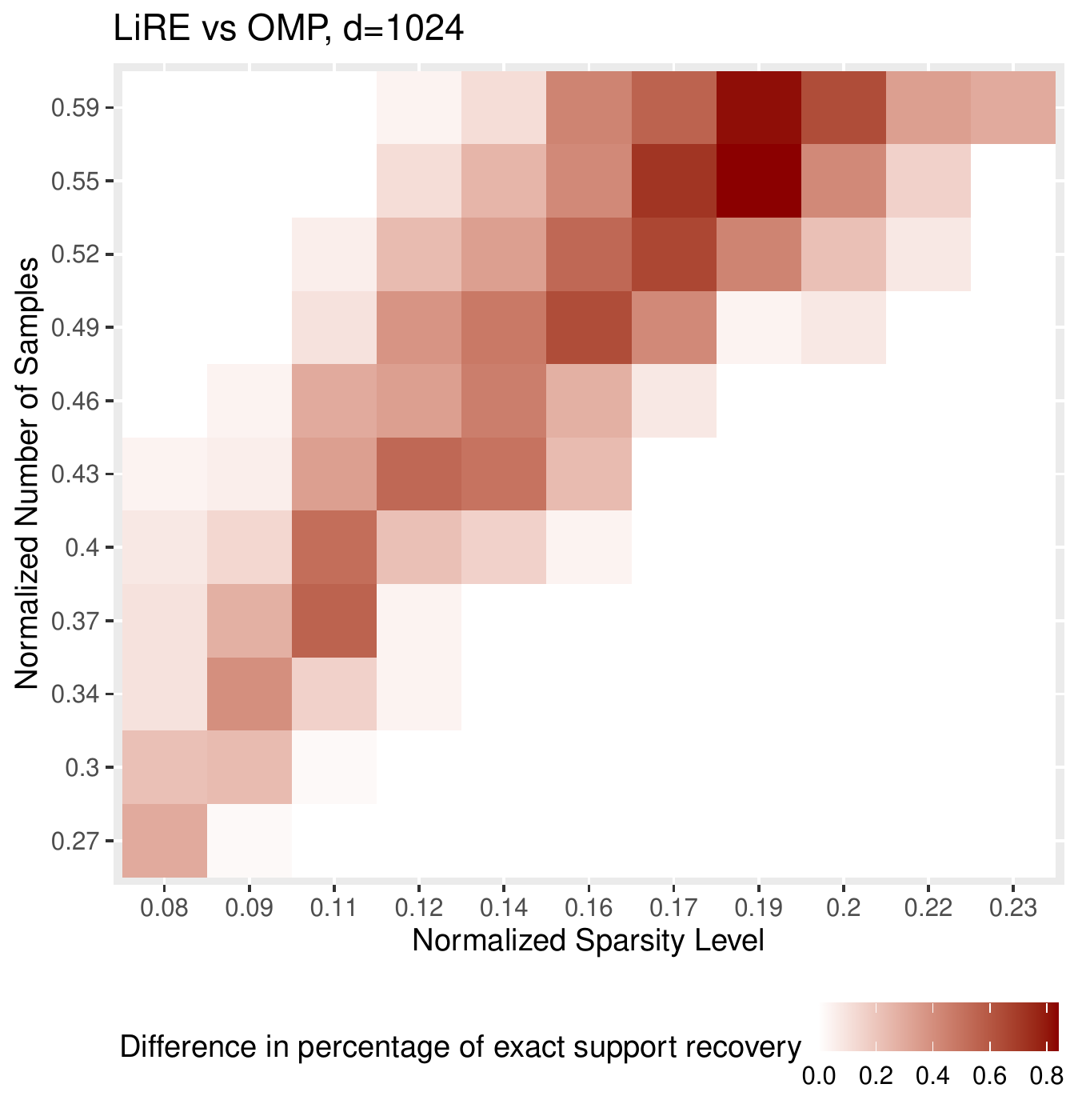}
  \caption{LiRE vs. OMP.}
  \label{fig:sfig2}
\end{subfigure}
\caption{LiRE as a standalone support recovery algorithm. Red squares favor BP in Fig.~\ref{fig:sfig1} and favor LiRE in Fig.~\ref{fig:sfig2}.}
\label{fig:fig}
\end{figure}

\section{Analysis} \label{canalysis}

\begin{lemma}[Lemmas $1$ and $2$ in \cite{1}]\label{lemma1}
Let $\bl,\bs$ be disjoint support vectors and suppose $\delta_{|\bl|+|\bs|}<1$.
Then,
for all $\bf{a} \in \mathbb{R}^{|\bl|}$ and $\bf{b}\in \mathbb{R}^{|\bs|}$ we have
\begin{align*}
   & |\bf{a}^\sT \bphi_{\bl}^\sT \bphi_{\bs}\bf{b}| \leq \delta_{|\bl|+|\bs|}||\bf{a}||_2||\bf{b}||_2\,,\\
   &||\bphi_{\bl}^\sT \bphi_{\bs}\bf{b}||_2 \leq \delta_{|\bl|+|\bs|}||\bf{b}||_2\,.
\end{align*}
\end{lemma}
\begin{lemma}\label{lemma2}
Let $\bs,\bq$ be support vectors that are disjoint from support vector $\bl$. If  $\max\{\delta_{|\bs|+|\bl|},\delta_{|\bq|+|\bl|}\}<1$, then for any ${\bf{a}} \in \mathbb{R}^{|\bs|}$ we have
\begin{align*}
 &||P\{\bphi_{\bl}\}\bphi_{\bs}{\bf{a}}||_2 \leq \frac{\delta_{|\bs|+|\bl|}}{\sqrt{1-\delta_{|\bl|}}}||{\bf{a}}||_2\,,\\
 &||\bphi_\bq^\sT  P\{\bphi_{\bl}\}\bphi_{\bs}{\bf{a}}||_2 \leq \frac{\delta_{|\bs|+|\bl|}\delta_{|\bq|+|\bl|}}{{1-\delta_{|\bl|}}}||\bf{a}||_2\,.
 \end{align*}
\end{lemma}
\begin{proof}
We know that $$||P\{\bphi_{\bl}\}\bphi_{\bs}{\bf{a}}||_2^2=\bu^\sT (\bphi^\sT _{\bl}\bphi_{\bl})^{-1}\bu\,,$$ where $\bu\defeq\bphi_{\bl}^\sT \bphi_{\bs}{\bf{a}}$. From the definition of the restricted isometry property we get $$\bu^\sT (\bphi^\sT _{\bl}\bphi_{\bl})^{-1}\bu \leq \frac{||\bu||_2^2}{1-\delta_{|\bl|}}\,. $$ Next, Lemma~\ref{lemma1} gives us $$||\bu||_2 \leq \delta_{|\bl|+|\bs|}||\ba||_2\,.$$ This implies $$||P\{\bphi_{\bl}\}\bphi_{\bs}a||_2 \leq \frac{\delta_{|\bl|+|\bs|}}{\sqrt{1-\delta_{|\bl|}}}||a||_2. $$ For the proof of the second inequality in Lemma~\ref{lemma2} first assume $||\bb||_2\leq 1$. Then, the Cauchy inequality yields $$\bb^\sT \bphi_\bq^\sT  P\{\bphi_{\bl}\}\bphi_{\bs}\ba \leq ||P\{\bphi_{\bl}\}\bphi_\bq \bb||_2 || P\{\bphi_{\bl}\}\bphi_{\bs}\ba   ||_2.$$Next, by invoking the first inequality in  Lemma \ref{lemma2} we obtain $$\max\limits_{||\bb||_2 \leq 1}^{} \bb^\sT \bphi_\bq^\sT  P\{\bphi_{\bl}\}\bphi_{\bs}\ba \leq  \frac{\delta_{|\bs|+|\bl|}\delta_{|\bq|+|\bl|}}{{1-\delta_{|\bl|}}}||\ba||_2\,,$$
which establishes the desired result.
\end{proof}

\begin{lemma}\label{lemma4}
Let $\bs_1,\bs_2$ be two disjoint support vectors with $|\bs_1|=k_1$ and $|\bs_2|=k_2$. If $\delta_{k_1+k_2} \leq 0.5$, then $\bphi_{\bs_1}^\sT (I-P\{\bphi_{\bs_2}\})\bphi_{\bs_1}$ is invertible and 
\begin{align*}
    (\bphi_{\bs_1}^\sT (I-P\{\bphi_{\bs_2}\})\bphi_{\bs_1})^{-1} \succeq \frac{1}{1+\delta_{k_1}}I_{k_1}\,.
\end{align*}
\end{lemma}
\begin{proof}
Given $\bz \in \mathbb{R}^{k_1}$ we have $$\bz^\sT (\bphi_{\bs_1}^\sT (I-P\{\bphi_{\bs_2}\})\bphi_{\bs_1})z=||\bphi_{\bs_1}\bz||_2^2 - ||P\{\bphi_{\bs_2}\}\bphi_{\bs_1}\bz||_2^2\,,$$ which is no larger than $ ||\bphi_{\bs_1}\bz||_2^2$. Next, recall that the RIP property implies 
\begin{equation}\label{eq:dummy2}
||\bphi_{\bs_1}z||_2^2\leq (1+\delta_{k_1})||\bz||_2^2.
\end{equation}

Therefore,   
$$(\bphi_{\bs_1}^\sT (I-P\{\bphi_{\bs_2}\})\bphi_{\bs_1}) \preceq (1+\delta_{k_1})I_{k_1}.$$
Moreover, Lemma \ref{lemma2} gives us  
\begin{equation}\label{eq:dummy3}
||P\{\bphi_{\bs_2}\}\bphi_{\bs_1}\bz||_2^2\leq \frac{\delta_k^2||\bz||_2^2}{1-\delta_{k_2}}.
\end{equation}
By combining~\eqref{eq:dummy2} and~\eqref{eq:dummy3} we get 
\begin{align*}(1-\delta_{k_1}-\frac{\delta_{k}^2}{1-\delta_{k_2}})I_{k_1} &\preceq(\bphi_{\bs_1}^\sT (I-P\{\bphi_{\bs_2}\})\bphi_{\bs_1}) \nonumber\\
&\preceq (1+\delta_{k_1})I_{k_1}\,.  \end{align*} Therefore $(\bphi_{\bs_1}^\sT (I-P\{\bphi_{\bs_2}\})\bphi_{\bs_1})$ is invertible and we have the desired property.
 \end{proof}
\begin{lemma}\label{lemma3}
Let $\bs_1, \bs_2$ be two disjoint support vectors, let $\bs=\bs_1\cup\bs_2$, and suppose $\delta_{|\bs_1|+|\bs_2|}\leq 0.5$. Then,
 \begin{align*}
 \bx_{\bs_1}=(\bphi_{\bs_1}^\sT (I-P\{\bphi_{\bs_2}\})\bphi_{\bs_1})^{-1}\bphi_{\bs_1}^\sT \by^{\perp\bs_2}.  
\end{align*}
\end{lemma}
\begin{proof}
Since $\delta_{|\bs_1|+|\bs_2|} <1 $ we have that $\bphi_{\bs}$ is invertible and  $ {\bx}_{\bs}=(\bphi_{\bs}^\sT \bphi_{\bs})^{-1}\bphi_{\bs}^\sT \by$. Also, since $\bs=\bs_1 \cup \bs_2$ we have 
\begin{align*}
\begin{bmatrix}
 {\bx}_{\bs_1}\\
 {\bx}_{\bs_2}
\end{bmatrix}
= \begin{bmatrix}
\bphi_{\bs_1}^\sT \bphi_{\bs_1} & \bphi_{\bs_1}^\sT \bphi_{\bs_2}\\
\bphi_{\bs_2}^\sT \bphi_{\bs_1} & \bphi_{\bs_2}^\sT \bphi_{\bs_2}\\
\end{bmatrix}^{-1}
\begin{bmatrix}
\bphi_{\bs_1}^\sT \by\\
\bphi_{\bs_2}^\sT \by\\
\end{bmatrix}\,.
\end{align*}
Let $\bpsi$ denote the Schur complement of block $\bphi_{\bs_2}^\sT \bphi_{\bs_2}$ of the matrix $\bphi_{\bs}^\sT \bphi_{\bs}$. It then follows that  $$\bpsi=\bphi_{\bs_1}^\sT \bphi_{\bs_1}-\bphi_{\bs_1}^\sT \bphi_{\bs_2}
(\bphi_{\bs_2}^\sT \bphi_{\bs_2})^{-1}\bphi_{\bs_2}^\sT \bphi_{\bs_1},$${\it{i.e.}}, $\bpsi=\bphi_{\bs_1}^\sT (I-P\{\bphi_{\bs_2}\})\bphi_{\bs_1}$. Since $\delta_{|\bs_1|+|\bs_2|} \leq 0.5 $,  we have that $\bpsi$ is invertible by Lemma~\ref{lemma4}. On the other hand, from the properties of invertible Schur complement we know that $$ {\bx}_{\bs_1}=\bpsi^{-1}\bphi_{\bs_1}^\sT \by-\bpsi^{-1}\bphi_{\bs_1}^\sT \bphi_{\bs_2}(\bphi_{\bs_2}^\sT \bphi_{\bs_2})^{-1}\bphi_{\bs_2}^\sT \by,$$ which can be rewritten as ${\bx}_{\bs_1}=\bpsi^{-1}\bphi_{\bs_1}^\sT \lsr{\by}{\bs_2}$, thereby completing the proof.  
\end{proof}
\begin{lemma}\label{lemma5}
Fix a length $k$ support vector $\bl$, let $\ell\defeq \max\{m+1,k\}$, and let $\bx$ be the regression of $\by$ on $\bl$, {\it{i.e.}},  $\bx=\cE(\bl)$. 
If $\delta_{\ell}\leq 0.5$, then for any $1\leq j \leq k$ we have 
\begin{equation*}
|\bx_{\bl[j]}| \leq \frac{\delta_\ell}{1-\delta_{\ell}-\delta_{\ell}^2}||\bxs_{\bss/\bl}||_2
\end{equation*}
whenever $\bl[j] \notin \bss$.
\end{lemma} 
 \begin{proof}
 Write $\bl$ as $\bl[j]\cup\bl[-j]$. Since $\delta_{k}\leq 0.5$ we can apply Lemma \ref{lemma3} and get 
 \begin{equation}\label{eq:dummy6}
 |\bx_{\bl[j]}|=\frac{|\bphi_{\bl[j]}^\sT \by^{\perp\bl[-j] }|}{||\bphi_{\bl[j]}^{\perp \bl[-j]}||_2^2}.
 \end{equation}
  We proceed by upper-bounding the numerator and lower-bounding the denominator of~\eqref{eq:dummy6}. For the numerator we have
 \begin{align*}
     |\bphi_{\bl[j]}^\sT \by^{\perp\bl[-j] }|&=|\bphi_{\bl[j]}^\sT (I-P\{\bphi_{\bl[-j]}\})\bphi_{\bss}\bxs| \\
     &=|\bphi_{\bl[j]}^\sT (I-P\{\bphi_{\bl[-j]}\})\bphi_{\bss/\bl[-j]}\bxs_{\bss/\bl[-j]}|\\
     &\leq |\bphi_{\bl[j]}^\sT \bphi_{\bss/\bl[-j]}\bxs_{\bss/\bl[-j]}| +|\bphi_{\bl[j]}^\sT P\{\bphi_{\bl[-j]}\}\bphi_{\bss/\bl[-j]}\bxs_{\bss/\bl[-j]}|\\
     &\overset{(a)}{\leq} \delta_{m+1}||\bxs_{\bss/\bl[-j]}||_2+||P\{\bphi_{\bl[-j]}\}\bphi_{\bl[j]}||_2||P\{\bphi_{\bl[-j]}\}\bphi_{\bss/\bl[-j]}\bxs_{\bss/\bl[-j]}||_2\\
     &\overset{(b)}{\leq}\delta_{m+1}||\bxs_{\bss/\bl[-j]}||_2 + 
     \frac{\delta_{k}^2}{1-\delta_{k-1}}||\bxs_{\bss/\bl[-j]}||_2\,,
     \end{align*}
     where $(a)$ follows from Lemma~\ref{lemma1} and where $(b)$ holds by Lemma~\ref{lemma2} and the Cauchy inequality. Therefore, from the monotonicity of the RIP constant we get
     \begin{equation}\label{eq:dummy4}
     |\bphi_{\bl[j]}^\sT \by^{\perp\bl[-j] }| \leq \frac{\delta_{\ell}}{1-\delta_{\ell}}||\bxs_{\bss/\bl[-j]}||_2.
     \end{equation}
     For the denominator of~\eqref{eq:dummy6} we have $$||\bphi_{\bl[j]}^{\perp \bl[-j]}||_2^2=1-||P\{\bphi_{\bl[-j]}\}\bphi_{\bl[j]}||_2^2,$$ hence by Lemma~\ref{lemma2} we get  $||\bphi_{\bl[j]}^{\perp \bl[-j]}||_2^2 \geq 1-\frac{\delta^2_k}{1-\delta_{k-1}}$. This inequality together with inequalities~\eqref{eq:dummy4} and~\eqref{eq:dummy6} complete the proof. 
 \end{proof}

\begin{lemma}\label{lemma6}
Let $\bs$ be a support vector of size $k$, and $e=|\bss\backslash\bs|$. Under setting~\eqref{compsens}, suppose $k+\ell > m$, $\ell \leq e $, and suppose the design matrix $\bphi$ satisfies $\delta_{k+\ell} \leq 0.5$.  Let $\bl\defeq  \cL(\ell,\bs)$,
$ {\bx}\defeq\cE(\bp)$, and $t\defeq \max\{k+e,\ell+e\}$. Consider $j \in \bl$ such that $|\hbx_j|=||\hbx_{\bl}||_{\infty}$. If $\bss \not\subset\bp$ and $j \notin \bss$, then
$$ ||\bxs_{\bl\cap\bss}||_2 \leq \left(\frac{\delta_{t}(1-\delta^2_{t})\sqrt{2e}}{(1-\delta_t-\delta_t^2)(1-2\delta_t)}-1\right) ||\bxs_{\bss\backslash\bp}||_2.$$
\end{lemma}
\begin{proof}
From Lemma~\ref{lemma5} we have 
\begin{equation}\label{eq:dummy7}
|\hbx_j|\leq \frac{\delta_{k+\ell}}{1-\delta_{k+\ell}-\delta_{k+\ell}^2} ||\bxs_{\bss\backslash\bp}||_2.
\end{equation}
On the other hand, we have
  \begin{align*}
    ||\hbx_{\bl}||_{\infty}&=||(\bphi_{\bl}^\sT (I-P\{\bphi_{\bs}\})\bphi_{\bl})^{-1}\bphi_{\bl}^\sT \by^{\perp \bs}||_{\infty}\\
    &\geq \frac{1}{\sqrt{\ell}}
    ||(\bphi_{\bl}^\sT (I-P\{\bphi_{\bs}\})\bphi_{\bl})^{-1}\bphi_{\bl}^\sT \by^{\perp \bs}||_2 \\
    &\overset{(a)}{\geq} \frac{1}{\sqrt{\ell}(1+\delta_{\ell})}||\bphi_{\bl}^\sT \by^{\perp \bs}||_2 \\
    &\overset{(b)}{\geq} \frac{1}{\sqrt{e}(1+\delta_{\ell})} || \bphi_{\bss\backslash\bs}^\sT \by^{\perp \bs}||_2 \\
    &\overset{(c)}{\geq} \frac{1}{\sqrt{e}(1+\delta_{\ell})}\big( ||\bphi_{\bss\backslash\bs}^\sT \bphi_{\bss\backslash\bs}
    \bxs_{\bss\backslash \bs}||_2-\\
    &||\bphi_{\bss\backslash\bs}^\sT P\{\bphi_{\bs}\}\bphi_{\bss\backslash\bs} \bxs_{\bss\backslash \bs}||_2\big)\\
    &\overset{(d)}{\geq} \frac{1}{\sqrt{e}(1+\delta_{\ell})} \left((1-\delta_e)||\bxs_{\bss\backslash\bs}||_2-\frac{\delta^2_{k+e}}{1-\delta_{k}}||\bxs_{\bss\backslash\bs}||_2\right)\,,
    \end{align*}
where $(a)$ follows from Lemma \ref{lemma4}; where $(b)$ is a direct result from the definition of $\bl$ along with the fact that $\bss\backslash \bp$ is non-empty and $|\bl|\leq e$; where $(c)$ follows from definition of $\lsr{\by}{\bs}$ and the triangle inequality; and where $(d)$ follows from the RIP definition together with Lemma~\ref{lemma2}. Next, by combining the above inequality with~\eqref{eq:dummy7} and using the monotonicity of the RIP constant, we deduce that if $j \notin \bss$, then 
$$\frac{\delta_t}{1-\delta_t-\delta_t^2}||\bxs_{\bss\backslash\bp}||_2 \geq \frac{(1-2\delta_t)}{\sqrt{e}(1-\delta_t^2)}||\bxs_{\bss\backslash\bs}||_2\,.$$
This implies 
\begin{equation}\label{eq:dummy5}
||\bxs_{\bss\backslash\bs}||_2 \leq \frac{\delta_{t}(1-\delta^2_{t})\sqrt{e}}{(1-\delta_t-\delta_t^2)(1-2\delta_t)}||\bxs_{\bss\backslash\bp}||_2 .
\end{equation}
Now, since $\bss\backslash\bs=(\bl\cap\bss) \cup (\bss\backslash\bp)$, we have $\sqrt{2}||\bxs_{\bss\backslash\bs}||_2\geq||\bxs_{\bl\cap\bss}||_2+||\bxs_{\bss\backslash\bp}||_2 $. Using this inequality in~\eqref{eq:dummy5} completes the proof.
\end{proof}

\begin{lemma}\label{lemma7}
Let $\bs$ be a support vector of size $k$. Let $\bl\defeq \cL(\ell,\bs)$,  $\bp\defeq \bs\cup\bl$, $e\defeq |\bss \backslash \bs|$, and $t\defeq \max\{k+e,e+\ell \}$. If $\bss \not\subset \bp$, then 
\begin{align*}
    &\left( \frac{1+\delta_t-\delta_{t}^2}{1-\delta_{t}}\right)||\bxs_{\bl\cap\bss}||_2 + \left( \frac{2\delta_{t}-\delta_{t}^2}{1-\delta_{k}}\right)||\bxs_{\bss\backslash\bp}||_2 \geq \Big( \frac{1-2\delta_{t}}{1-\delta_t}||\bxs_{\bss\backslash\bp}||_2-
 \frac{\delta_t}{1-\delta_t}||\bxs_{\bss\cap\bl}||_2 \Big)\cdot\sqrt{\frac{\ell}{{|\bss\backslash \bp|}}}
\end{align*}
\end{lemma}
\begin{proof}[Proof of Lemma \ref{lemma7}]
From the definition of $\cL(\ell,\bs)$ and the fact that $\bss\backslash \bp$ is non-empty, we have $\frac{||\bphi_{\bl}^\sT \lsr{\by}{\bs}||_2}{\sqrt{\ell}} \geq ||\bphi_{\bss\backslash \bp}^\sT \lsr{\by}{\bs}||_{\infty} $. Then, since $||\bphi_{\bss\backslash \bp}^\sT \lsr{\by}{\bs}||_{\infty}\geq \frac{||\bphi_{\bss\backslash \bp}^\sT \lsr{\by}{\bs}||_{2}}{\sqrt{| \bss\backslash \bp|}},$ we get 
\begin{equation}\label{eq3}
||\bphi_{\bl}^\sT \lsr{\by}{\bs}||_2 \geq \sqrt{\frac{\ell}{{|\bss\backslash \bp|}}}||\bphi_{\bss\backslash \bp}^\sT \lsr{\by}{\bs}||_{2}.\end{equation}
Now, on the hand we have 
\begin{align}
||\bphi_{\bss\backslash \bp}^\sT \lsr{\by}{\bs}||_{2}&=||\bphi_{\bss\backslash \bp}^\sT (I-P\{\bphi_{\bs}\})\bphi_{\bss\backslash \bs}\bxs_{\bss\backslash \bs}||_{2}\notag\\
&\geq ||\bphi_{\bss\backslash \bp}^\sT \bphi_{\bss\backslash \bs}\bxs_{\bss\backslash \bs}||_{2} -||\bphi_{\bss\backslash \bp}^\sT P\{\bphi_{\bs}\}\bphi_{\bss\backslash \bs}\bxs_{\bss\backslash \bs} ||_{2}\notag\\
&\overset{(a)}{\geq} ||\bphi_{\bss\backslash \bp}^\sT \bphi_{\bss\backslash \bs}\bxs_{\bss\backslash \bs}||_{2} - \frac{\delta_{k+e}^2}{1-\delta_k}||\bxs_{\bss\backslash \bs}||_2\notag \\ 
&\overset{(b)}{\geq} ||\bphi_{\bss\backslash \bp}^\sT \bphi_{\bss\backslash \bp}\bxs_{\bss\backslash \bp}||_{2}- ||\bphi_{\bss\backslash \bp}^\sT \bphi_{\bss\cap\bl}\bxs_{\bss\cap\bl}||_{2}-  \frac{\delta^2_{k+e}}{1-\delta_k} ||\bxs_{\bss\backslash \bs}||_2 \notag\\
&\overset{(c)}{\geq} ||\bphi_{\bss\backslash \bp}^\sT \bphi_{\bss\backslash \bp}\bxs_{\bss\backslash \bp}||_{2}- \delta_{e}||\bxs_{\bss\cap\bl}||_{2}-
\frac{\delta^2_{k+e}}{1-\delta_k} ||\bxs_{\bss\backslash \bs}||_2 \notag\\ 
&\overset{(d)}{\geq} (1-\delta_{|\bss\backslash \bp|})||\bxs_{\bss\backslash \bp}||_2-\delta_{e}||\bxs_{\bss\cap\bl}||_{2}- \frac{\delta^2_{k+e}}{1-\delta_k} ||\bxs_{\bss\backslash \bs}||_2 \notag\\
&\overset{(e)}{\geq} (1-\delta_{e})||\bxs_{\bss\backslash \bp}||_2-\delta_{e}||\bxs_{\bss\cap\bl}||_{2}- \frac{\delta^2_{k+e}}{1-\delta_k} ||\bxs_{\bss\backslash \bs}||_2 \notag\\
&\overset{(b)}{\geq} (1-\delta_e-\frac{\delta^2_{k+e}}{1-\delta_k})||\bxs_{\bss\backslash \bp}||_2 -
(\delta_e+ \frac{\delta_{k+e}^2}{1-\delta_k})||\bxs_{\bss\cap\bl}||_2\,, \label{eq1}
\end{align}
where $(a)$ holds by Lemma~\ref{lemma2}; where $(b)$ follows from the triangle inequality and the identity $\bss\backslash \bs=(\bss\cap\bl) \cup (\bss\backslash \bp)$; where $(c)$ is a direct result of Lemma~\ref{lemma1}; where $(d)$ follows from the definition of RIP constant; and where $(e)$ holds by the monotonicity of the RIP constant and the fact that $|\bss\backslash \bp| \leq e$. 

On the other hand, we have 
\begin{align}
    ||\bphi_{\bl}^\sT \lsr{\by}{\bs}||_2&=||\bphi_{\bl}^\sT (I-P\{\bphi_{\bs}\})\bphi_{\bss}\bxs_{\bss}||_2 \notag \\
    &=||\bphi_{\bl}^\sT (I-P\{\bphi_{\bs}\})\bphi_{\bss\backslash \bs}\bxs_{\bss\backslash \bs}||_2\notag \\
    &\leq ||\bphi_{\bl}^\sT \bphi_{\bss\backslash \bs}\bxs_{\bss\backslash \bs}||_2 + 
    ||\bphi_{\bl}^\sT P\{\bphi_{\bs}\}\bphi_{\bss\backslash \bs}\bxs_{\bss\backslash \bs}||_2\notag \\
    &\overset{(a)}{\leq} ||\bphi_{\bl}^\sT \bphi_{\bss\backslash \bs}\bxs_{\bss\backslash \bs}||_2 + \frac{\delta_{e+k}\delta_{k+\ell}}{1-\delta_{k}}||\bxs_{\bss\backslash \bs}||_2 \notag \\
    &\overset{(b)}{\leq} ||\bphi_{\bl\backslash \bss}^\sT \bphi_{\bss\backslash \bs}\bxs_{\bss\backslash \bs}||_2 + ||\bphi_{\bl\cap\bss}^\sT \bphi_{\bss\backslash \bs}\bxs_{\bss\backslash \bs}||_2+
    \frac{\delta_{e+k}\delta_{k+\ell}}{1-\delta_{k}}||\bxs_{\bss\backslash \bs}||_2\notag \\
    &\overset{(c)}{\leq}(\delta_{\ell+e}+\frac{\delta_{e+k}\delta_{k+\ell}}{1-\delta_{k}})||\bxs_{\bss\backslash \bs}||_2+||\bphi_{\bl\cap\bss}^\sT \bphi_{\bss\backslash \bs}\bxs_{\bss\backslash \bs}||_2\notag \\
    & \overset{(d)}{\leq} (\delta_{\ell+e}+\frac{\delta_{e+k}\delta_{k+\ell}}{1-\delta_{k}})||\bxs_{\bss\backslash \bs}||_2 + ||\bphi_{\bl\cap\bss}^\sT \bphi_{\bl\cap\bss}\bxs_{\bl\cap\bss}||_2 +||\bphi_{\bl\cap\bss}^\sT \bphi_{\bss\backslash \bp}\bxs_{\bss\backslash \bp}||_2\notag \\
    & \overset{(e)}{\leq} (\delta_{\ell+e}+\frac{\delta_{e+k}\delta_{k+\ell}}{1-\delta_{k}})||\bxs_{\bss\backslash \bs}||_2  + (1+\delta_e)||\bxs_{\bl\cap \bss}||_2 +||\bphi_{\bl\cap\bss}^\sT \bphi_{\bss\backslash \bp}\bxs_{\bss\backslash \bp}||_2 \notag\\
    &\overset{(f)}{\leq} (\delta_{\ell+e}+\frac{\delta_{e+k}\delta_{k+\ell}}{1-\delta_{k}})||\bxs_{\bss\backslash \bs}||_2  + (1+\delta_e)||\bxs_{\bl\cap \bss}||_2 +\delta_e || \bxs_{\bss\backslash \bp} ||_2\notag \\
    &\leq (1+\delta_e+\delta_{\ell+e}+ \frac{\delta_{e+k}\delta_{k+\ell}}{1-\delta_{k}})||\bxs_{\bl\cap\bss}||_2 + (\delta_e+\delta_{\ell+e}+ \frac{\delta_{e+k}\delta_{k+\ell}}{1-\delta_{k}})||\bxs_{\bss\backslash \bp}||_2\,, \label{eq2}
    \end{align}
where $(a)$ follows from Lemma~\ref{lemma2}; where $(b)$ follows from the triangle inequality applied to identity $\bl=(\bl\backslash \bss)\cup (\bl\cap \bss)$; where $(c)$ follows from Lemma~\ref{lemma1} and from the fact that $|\bl\backslash \bss|\leq \ell$; where $(d)$ follows from the triangle inequality applied to identity $\bss\backslash \bs=(\bss\backslash \bp) \cup (\bss\cap\bl)$; where $(e)$ follows from the definition of RIP and the fact that $|\bss\cap\bl|\leq e$; and where $(f)$ follows from Lemma~\ref{lemma1} and the fact that $e=|\bss\cap\bl|+|\bss\backslash \bp|$. 

Finally, combine \eqref{eq1}, \eqref{eq2}, and \eqref{eq3} along with $t=\max\{k+e,\ell+e\}$ and the monotonicity of the RIP constant to get the desired result.  
\end{proof}
\begin{proposition}\label{prop1}
Consider the setting in \eqref{compsens}. Let $\bs$ be a support vector of size $k$ and let $e\defeq|\bss\backslash \bs|$. For an integer $\ell$ that satisfies $\ell\leq e$ and $k+\ell>m$ let $\bl\defeq \cL(\ell,\bs )$, $\bp\defeq \bl\cup\bs$,
$\hat{\bx}\defeq \cE(\bp)$, and $t\defeq \max\{k+e,\ell+e \}.$ Suppose $j \in \bl$ is such that $|\hbx_j|=||\hbx_{\bl}||_{\infty}$. Furthermore, assume $\delta_{k+\ell} \leq 0.5$ and $\bss \not\subset\bp$. If the two inequalities
$$\sqrt{e} \leq \frac{\sqrt{2}(1-\delta_t-\delta_t^2)(1-2\delta_t) }{(1-\delta_t^2+2\delta_t)\delta_t(1+\delta_t)} \,,$$
and
$$\sqrt{\ell} > \frac{(1-\delta_t^2+\delta_t)(\eta\sqrt{e}-1)+2\delta_t-\delta_t^2}{1-2\delta_t-\delta_t(\eta\sqrt{e}-1)}\sqrt{e}\,$$
hold
with $\eta\defeq\frac{(1-\delta_t^2)\delta_t\sqrt{2}}{(1-\delta_t-\delta_t^2)(1-2\delta_t)}$, then $j \in \bss$.
\end{proposition}
\begin{proof}
Assume $j\notin \bss $. Then, from Lemma~\ref{lemma6} we get  
\begin{equation}\label{eq:dummy8}
||\bxs_{\bl\cap\bss}||_2 \leq {(\eta\sqrt{e}-1})||\bxs_{\bss\backslash \bp}||_2.
\end{equation}
From~\eqref{eq:dummy8} and Lemma~\ref{lemma7} we get 
\begin{align*}
    &||\bxs_{\bss\backslash \bp}||_2\left( \frac{1+\delta_t-\delta_t^2}{1-\delta_t}(\eta\sqrt{e}-1)+\frac{2\delta_t-\delta_t^2}{1-\delta} \right) \geq ||\bxs_{\bss\backslash \bp}||_2\left( \frac{1-2\delta_t}{1-\delta_t}-\frac{\delta_t}{1-\delta_t}(\eta\sqrt{e}-1)\right)\cdot\sqrt{\frac{\ell}{|\bss\backslash \bp|}}\,.
\end{align*}
Next, since $1-2\delta_t \geq \delta_t(\eta\sqrt{e}-1)$ and $|\bss\backslash\bp| \leq e$ we get
\begin{align*}
   &\left(\frac{1+\delta_t-\delta_t^2}{1-\delta_t}(\eta\sqrt{e}-1)+\frac{2\delta_t-\delta_t^2}{1-\delta} \right) \geq \left( \frac{1-2\delta_t}{1-\delta_t}-\frac{\delta_t}{1-\delta_t}(\eta\sqrt{e}-1)\right)\sqrt{\frac{\ell}{e}}\,,
\end{align*}
which contradicts the assumed lower bound on $\ell$. Therefore the initial assumption cannot hold and $j$ must belong to $\bss$.
\end{proof}

\begin{proof}[Proof of Theorem\ref{thm1}]
We start with the first step of LiRE ($i=1$). Let $\bs\defeq\bs_{\text{out}}[-1]$, $\bl\defeq \cL(\ell,\bs), \bp\defeq \bl\cup\bs,$ and  $\hat{\bx}\defeq\cE(\bp)$. Pick $j \in \bl$ such that $|\hat{\bx}_j|=||\bx_{\bl}||_{\infty}$. Note that since $|\bss\backslash \bs|\leq |\bss\backslash \bs_{\text{out}}|+1$, we have  $|\bss\backslash \bs|\leq e+1$. We have the following three possible cases:

\noindent \textbf{Case 1:} $\bss\not\subset\bp$. Here the assumptions of Proposition ~\ref{prop1} are satisfied and therefore $j \in \bss$. This implies $\bss_{\text{out}}[1] \in \bss$, hence the first component of the estimated support is corrected.

\noindent \textbf{Case 2:} $\bss\subset\bp$ and $\bss \subset \bs_{\text{out}}[-1]$.
Since $\bss \subset \bs_{\text{out}}[-1]$, we have $\by^{\perp \bs_{\text{out}}[-1]}=0$ and LiRE will exit the {\bf{for}} loop and output $\bs_{\text{out}}$. Moreover, we get  $\bss \subset \bs_{\text{out}} $ which completes the proof.

\noindent\textbf{Case 3:} $\bss \subset \bp$ and $\bss \not\subset \bs_{\text{out}}[-1]$. The next lemma characterizes any true feature in the support estimate. 

\begin{lemma}\label{lemma8}
 If $\bss \subset \bp$ and $\bss \not\subset \bs_{\text{out}}[-1]$, then for every $i$ in $\bl$ we have $i\in \bss$ if and only if $|\hat{\bx}_{i}|>0$.
\end{lemma}
\noindent We have $\bss\subset{\bs_{\text{out}}[-1] \cup \bl}$ and since $\bss \not\subset \bs_{\text{out}}[-1]$ we deduce that there must exist some feature $i\in \bl\cap \bss$. By Lemma~\ref{lemma8}, we get $|\hat{\bx}_i|>0$. Moreover, the assumption that $|\hat{\bx}_j|=||\hat{\bx}_{\bl}||_{\infty}$ implies $|\hat{\bx}_j|\geq |\hat{\bx}_i|$, hence $|\hat{\bx}_j| >0$. From Lemma~\ref{lemma8} we then conclude that $j \in \bss$. This means that LiRE will add a true feature in this case as well.


The above argument holds for the first iteration of the {\bf{for}} loop of LiRE. However, as we are not adding errors to $\bs_{\text{out}}$, the same argument can be used for the other iterations as well. As there are at most $m$ missed true features outside $\bs_{\text{in}}$, after one round all of them will be included and therefore $\bss \subset \bs_{\text{out}}$. 
\begin{proof}[Proof of Lemma \ref{lemma8}]
 Pick $i \in \bl$. If $i \notin \bss$, next since $\delta_{m+\ell}<0.5$, we can apply Lemma \ref{lemma5} and get $|\hat{\bx}_{i}|\leq \frac{\delta_{m+\ell}}{1-\delta_{m+\ell}-\delta_{m+\ell}^2}||\bxs_{\bss\backslash \bp}||_2$, then this implies $|\hat{\bx}_i|=0$. It remains to show the converse, namely that if $i \in \bss\cap \bl$ then $|\hat{\bx}_i|>0$. Suppose that there exists $i \in \bss\cap\bl$ such that $|\hat{\bx}_i|=0$. From Lemma~\ref{lemma4} we get  
 \begin{equation}\label{eq:dummy9}
 |\hat{\bx}_i|=\frac{|\bphi_i^\sT \lsr{\by}{\bp/\{i\}}|}{||\lsr{\bphi_i}{\bp/\{i\}}||_2^2}\,,
 \end{equation}
which yields ${|\bphi_i^\sT \lsr{\by}{\bp/\{i\}}|}=0\,.$ Next, by expanding~\eqref{eq:dummy9} we get
$$|\bphi_i^\sT (I-P\{\bphi_{\bp/\{i\}}\})\bphi_{\bss}\bxs_{\bss}|=0.$$ Then, since $\bss \subset \bp$ and $i \in \bss$, we deduce that 
\begin{equation}\label{eq: dummy1}
|\bphi_i^\sT (I-P\{\bphi_{\bp/\{i\}}\})\bphi_{i}\bxs_{i}|=0.
\end{equation}
On the other hand, we have
\begin{align*}
    |\bphi_i^\sT (I-P\{\bphi_{\bp/\{i\}}\})\bphi_{i}\bxs_{i}| 
    &\geq |\bphi_i^\sT \bphi_{i}\bxs_{i}|-|\bphi_i^\sT P\{\bphi_{\bp/\{i\}}\}\bphi_{i}\bxs_{i}|\\
    &\overset{(a)}{\geq} |\bxs_i|-\frac{\delta_{m+\ell}^2}{1-\delta_{m+\ell}}|\bxs_i|\,,
    \end{align*}
where for $(a)$ we used Lemma~\ref{lemma2}. This implies $ |\bphi_i^\sT (I-P\{\bphi_{\bp/\{i\}}\})\bphi_{i}\bxs_{i}|>0$ and contradicts~\eqref{eq: dummy1}. Hence, if $i \in \bl\cap \bss $, then $|\hat{\bx_i}|>0$, which completes the proof.
\end{proof}
\end{proof}

\section{Concluding Remarks}\label{crem}
We proposed LiRE, a low complexity error-correction module for sparse recovery algorithms, and provided sufficient conditions under which LiRE corrects all errors made by the baseline algorithm. Simulations show that LiRE may boost the performance of low complexity greedy algorithms to attain the performance of significantly more complex ones, as we saw in the comparison of LiRE$\circ$OMP vs. BP. Alternatively, LiRE may be used as a fast standalone support recovery algorithm that is competitive against OMP.  Interesting venues for future research are in order:
\begin{itemize}
    \item 
Theorem~\ref{thm1} provides a sufficient condition for error-correction that appears conservative in light of the numerical experiments. For instance, Theorem~\ref{thm1} implies that if we want to correct up to $\bar{e}$ errors (as opposed to correcting exactly $e$ errors), then because of Condition~\eqref{elcond} the list size should be equal to one. This, in turn, yields Corollary~\ref{corollary2} which appears to be loose when the number of potential errors is  large, linear in $m$---see Section~\ref{stalone}. We also observe that the list size used for the numerical experiments appears to be a robust choice across sparsity levels. Improving Theorem~\ref{thm1} by potentially relaxing Condition~\eqref{elcond} is a natural direction for future investigation.
\item A very interesting problem is to quantify the successive refinement of the support estimate obtained by running LiRE multiple times.
\item LiRE keeps the list size constant throughout its iterations. However, as iterations proceed there are fewer and fewer errors to be corrected and since the list size impacts the likelihood of changing a feature, it might be possible to improve performance by considering an adaptive list size.

\item 
Theoretical guarantees in the noisy measurement setup would yield interesting and non-trivial extensions of the results presented in this paper.
\end{itemize}

 \section*{Acknowledgement}
 The authors would like to thank Dr. Venkat Chandar for his insightful comments.




\vskip 0.2in
\bibliographystyle{IEEEtran}
\bibliography{main}

\end{document}